\documentclass[12pt]{amsart}
\usepackage[T1]{fontenc}
\usepackage[a4paper,margin=0.8in]{geometry}
\usepackage{graphicx}
\usepackage{amssymb}
\usepackage{stmaryrd}
\usepackage{enumitem}
\usepackage{url}
\usepackage{mathtools}
\usepackage{amsmath,amsfonts,amsthm}
\usepackage{mathrsfs}

\usepackage{float}

\theoremstyle{plain}
\newtheorem{theorem}{Theorem}[section]
\newtheorem{lemma}[theorem]{Lemma}
\newtheorem{proposition}[theorem]{Proposition}
\newtheorem{corollary}[theorem]{Corollary}
\theoremstyle{definition}
\newtheorem{definition}[theorem]{Definition}
\newtheorem{example}[theorem]{Example}

\theoremstyle{remark}

\usepackage{multirow}
 
\newcommand{\ppi}{\mathbb{P}_s}

\usepackage[colorlinks,
linkcolor=blue,
anchorcolor=blue,
citecolor=blue
]{hyperref}

\author{J. Hamoud}
\address{\textbf{Jasem Hamoud:} Department of Discrete Mathematics, Moscow Institute of Physics and Technology}
\email{hamoud.math@gmail.com}
\thanks{Corresponding author. J. Hamoud (hamoud.math@gmail.com)}

\title[On Epimorphisms of Hypergraphic Automata]{On Epimorphisms of Hypergraphic Automata and Input Symbol Semigroups}

\date{}
\begin{document}

\begin{abstract}
Hypergraphic automata are automata whose state sets and output symbol sets are
hypergraphs invariant under the actions of the transition and output functions.
Universally attracting objects in the category of such automata are called
universal hypergraphic automata; their semigroups of input symbols are algebras
of mappings whose properties are tightly linked to the algebraic structure of
the automata themselves. In this paper we describe surjective morphisms
(epimorphisms) of universal hypergraphic automata and of their semigroups of
input symbols, for effective hypergraphs with $p$-definable edges. Under two mild hypotheses, a triple $(f,\mathbb P_s,g)$ is an
epimorphism if and only if the state component $f$ is an isomorphism and the output component $g$ is a homomorphism satisfying a lifting
condition. We introduce weak and strong epimorphisms of hypergraphs, show that they differ, and prove that liftability yields strong surjectivity under a size condition on
edges. For automata whose state hypergraphs are projective or affine planes,
the state components of epimorphisms are exactly the collineations. 
\end{abstract}

\maketitle

\noindent\rule{15.9cm}{1.0pt}


\medskip


\medskip

\noindent
{\bf
UDC:} 519.713+512.532+514.174

\noindent\rule{15.9cm}{1.0pt}

\section{Introduction}~\label{sec:intro}
Mappings are composed as presented by B.~Plotkin~\cite{plotkin}: for
$f\colon A\to B$ and $\varphi\colon B\to C$ the product is
$a(f\varphi)=(af)\varphi$ for every $a\in A$. From Section~\ref{sec001prelim}
on we write mappings on the left and compose them as functions. The identity
mapping of a set $A$ is $\varepsilon_A$, the set of all mappings from $A$ to $B$
is $\mathrm{Fun}(A,B)=B^{A}$, and a mapping $f^{-1}\colon B\to A$ is the
\emph{inverse} of $f\colon A\to B$ if $ff^{-1}=\varepsilon_A$ and
$f^{-1}f=\varepsilon_B$. Thus $f^{-1}$ exists if and only if $f$ is bijective.
A transformation of $A$ is a mapping $A\to A$; the set of all transformations
is denoted by $S_A$.

Following~\cite{plotkin}, an \emph{automaton} is a multi-sorted algebraic
system $\mathcal A=(X,S,Y,\delta,\lambda)$, where $X$ is a set of states, $Y$ is
a set of output symbols, $S$ is a semigroup of input symbols, and
$\delta\colon X\times S\to X$, $\lambda\colon X\times S\to Y$ are the transition
and output functions satisfying $\delta(x,ab)=\delta(\delta(x,a),b)$ and
$\lambda(x,ab)=\lambda(\delta(x,a),b)$. Depending on the application, the sets
$X$ and $Y$ may carry the structure of objects of a category $\mathbf K$ of
algebraic systems, with $\delta_s$ and $\lambda_s$ being morphisms of
$\mathbf K$ for every $s\in S$. For each pair $A,B\in\mathbf K$ the category of
automata structured over $\mathbf K$ has a universally attracting object, the
\emph{universal automaton} $\mathrm{Atm}(A,B)$, whose input symbol semigroup is
$S(A,B)=\mathrm{End}\,A\times\mathrm{Hom}(A,B)$~\cite{plotkin}. Every automaton
with state system $A$ and output system $B$ maps homomorphically into
$\mathrm{Atm}(A,B)$, so the algebraic structure of $S(A,B)$ governs the
behaviour of the whole class.

In this paper $\mathbf K$ is the category $\mathbf{Hgr}$ of effective
hypergraphs with $p$-definable edges. It contains, in particular, all
projective planes and all affine planes with more than four
points~\cite{karteszi,hartshorne}, as well as hypergraphs given by partitions
of the vertex set into non-singleton classes. For general background on
hypergraphs we refer to~\cite{bretto,zykov}, and for semigroups, relations and
languages of automata to~\cite{lallement,vagner,eilenberg}.

The study of universal hypergraphic automata by means of their input symbol
semigroups was started by Molchanov and Khvorostukhina. It is known that
$\mathrm{Atm}(H_X,H_Y)$ is determined up to isomorphism by $S(H_X,H_Y)$
\cite{molchanov2019}, that isomorphisms admit a concrete description
\cite{khvorostukhina2017}, and, most recently, that monomorphisms admit a
complete description~\cite{khvorostukhina2025}. The natural next question,
which is the subject of the present paper, is the description of epimorphisms.

Throughout, an \emph{epimorphism} means a \emph{surjective} homomorphism of
semigroups, hypergraphs or automata. This is not the categorical notion: for
instance, the inclusion $(\mathbb N,+)\hookrightarrow(\mathbb Z,+)$ is an
epimorphism in the category of semigroups but is not surjective. Hence all
statements below concern surjective morphisms.

For hypergraph homomorphisms the situation requires care, because surjectivity
on vertices does not control the edges. We therefore distinguish two notions
for a homomorphism $g\colon H_Y\to H_{Y_1}$: $g$ is a \emph{weak} epimorphism if
it is surjective on vertices, and a \emph{strong} epimorphism if, in addition,
every edge of $H_{Y_1}$ is the image of an edge of $H_Y$ (see
Definition~\ref{def:weak-strong}). The following small example shows that the
notions differ.

\begin{example}\label{ex0weaknotstrong}
Let $p=2$. Let $H_Y$ have vertex set $\{1,\dots,9\}$ and edges
$\{1,2,3\},\{4,5,6\},\{7,8,9\}$, and let $H_{Y_1}$ have the same vertices and
the same edges together with $\{1,4,7\}$. Both hypergraphs are effective with
$2$-definable edges (every edge has $3$ vertices and distinct edges meet in at
most one vertex). The identity map $g=\mathrm{id}$ is a homomorphism
$H_Y\to H_{Y_1}$ which is surjective on vertices, hence a weak epimorphism, but
the edge $\{1,4,7\}$ is not the image of any edge of $H_Y$, so $g$ is not a
strong epimorphism. Note also that $|L_{Y_1}|=4>|L_Y|=3$.
\end{example}

A second source of subtlety is the following. Given $\psi_1\in\mathrm{Hom}(H_X,H_{Y_1})$,
one wants $\psi\in\mathrm{Hom}(H_X,H_Y)$ with $g\circ\psi=\psi_1$. Surjectivity
of $g$ on vertices, or even on edges, does not suffice: a pointwise choice of
preimages need not be a homomorphism of hypergraphs
(Example~\ref{ex0notliftable}). Consequently, the description of epimorphisms of
$S(H_X,H_Y)$ requires an explicit \emph{lifting condition} on $g$
(Definition~\ref{def0lifting}).

Main results are:
\begin{enumerate}
\item We introduce weak and strong epimorphisms of hypergraphs and compare them
(Example~\ref{ex0weaknotstrong}, Proposition~\ref{prop0lift1strong}).
\item We describe when the pair $\ppi=(f^2,f\times g)$ is a well-defined
surjective homomorphism $S(H_X,H_Y)\to S(H_{X_1},H_{Y_1})$. Under two
hypotheses (V) and (P) (Theorem~\ref{thm0main}), the state component $f$ must
be a bijection (in fact an isomorphism) and the output component $g$ must
satisfy the lifting condition. Hypothesis (V) cannot be dropped
(Example~\ref{ex0V1needed}) and holds automatically in the geometric case
(Lemma~\ref{lem0V1auto}); the case $p=1$ is covered only for injective $f$.
\item We specialise the result to automata whose state hypergraphs are
projective or affine planes, where the isomorphisms are collineations
(Corollary~\ref{cor002}, Propositions~\ref{prop0PG} and~\ref{prop0AG}).
\item We show that epimorphisms transfer recognisable languages through
preimages (Proposition~\ref{prop0languages}).
\end{enumerate}

Section~\ref{sec:literature} reviews related work. Section~\ref{sec001prelim}
collects the definitions and the lifting property. Section~\ref{sec:main}
contains the main theorem. Section~\ref{sec:applications} gives consequences
and applications, and the final section is the Conclusion.

\section{Literature Review}\label{sec:literature}

In this section mappings are written on the left, as in the cited sources, and
the letters $Q$ and $\Gamma$ denote the state set and the input alphabet of an
automaton; they are local to this section and are not the objects $X$, $Y$,
$S$ of Section~\ref{sec:intro}.

A (complete, deterministic) automaton is a triple $A=(Q,\Gamma,\delta)$,
$\delta\colon Q\times \Gamma\to Q$, where $Q$ is the set of states and $\Gamma$
the finite input alphabet. The transition function extends uniquely to an
action of the free semigroup $\Gamma^{+}$ (equivalently the free monoid
$\Gamma^{*}$) on $Q$ by
\[
\delta(q,uv)=\delta(\delta(q,u),v),\qquad \delta(q,\varepsilon)=q .
\]
This free object is precisely the input symbol semigroup of the automaton. It
carries the congruence $\rho$ defined by
\[
u\,\rho\,v \iff \forall q\in Q:\ \delta(q,u)=\delta(q,v),
\]
and the quotient $c(A)=\Gamma^{+}/\rho$ is known as the characteristic
semigroup of $A$~\cite{yang1975}. Equivalently, writing
$f_w\colon Q\to Q$, $f_w(q)=\delta(q,w)$, the set
$F_A=\{f_w : w\in\Gamma^{+}\}$ with the product ``first $f_u$, then $f_v$'' is
the transformation semigroup of $A$, and $F_A\cong c(A)$. 

A function between state machines $M=(S,I,N)$, $M'=(S',I',N')$ is a
homomorphism if it consists of $\gamma\colon S\to S'$, $\iota\colon I\to I'$,
with $\gamma(N_1(s,i))=N_k(\gamma(s),\iota(i))$~\cite{yang1975}. If both
$\gamma$ and $\iota$ are onto, $M'$ is a homomorphic image; restricting to
$I=I'$ gives endomorphisms, automorphisms, etc.~\cite{yang1975}. This is the
``input-preserving'' notion; requiring $\iota$ to be a semigroup homomorphism
of the input symbol semigroups yields the strongest form.

Yang reduces the homomorphism problem to closed partitions and
output-consistent partitions: an endomorphism of a state machine exists iff
there is a non-trivial closed partition (Hartmanis--Stearns), and the
homomorphism, isomorphism and automorphism problems of Moore-type machines are
decided by closed, output-consistent partitions~\cite{yang1975}. Bavel
decomposes an arbitrary finite automaton into its primaries (maximal
singly-generated subautomata) and shows that every homomorphism is uniquely a
join of homomorphisms on those primaries: any $f\in H(A\to B)$ has a unique
decomposition $f=\bigvee_i f_i$, $f_i\in H(A_i\to B)$, over the primaries $A_i$
of $A$~\cite{bavel1968}.

\subsection{Hypergraph semigroups and graph-algebra automata}

A hypergraph is a triple $H=(V,E,I)$ with $I\subseteq V\times E$ the incidence
relation; ordinary graphs, multigraphs and designs are special cases.
starting from the groupoid $G(H)$ on $V\times E$ with
\[
(v_1,e_1)(v_2,e_2)=
\begin{cases}
(v_1,e_2), & (v_2,e_1)\in I,\\
\text{undefined}, & \text{otherwise,}
\end{cases}
\]
one adjoins a zero to obtain $S(H)$, which is isomorphic to the Rees matrix
semigroup $M^{0}(1;E,V;1)$ with trivial structure group. Two non-empty
hypergraphs $H_1,H_2$ are isomorphic iff $S(H_1)\cong S(H_2)$, and every
isomorphism is induced by a uniquely determined hypergraph isomorphism
$(\varphi,\psi)$. Abstractly, a semigroup $S$ is (isomorphic to) a hypergraph
semigroup iff it satisfies five conditions:
\begin{enumerate}
  \item $S$ is categorical at zero;
  \item $xSy$ contains at most one non-zero element;
  \item if $xSy=\{0\}$ then $xS=\{0\}$ or $Sy=\{0\}$;
  \item $st=uv=0$ forces $s=ux$, $v=xt$, for some $x$;
  \item $S^{2}=\{0\}$ or $S^{2}=S$.
\end{enumerate}

A graph-algebra automaton is built by using, for $D=(V,E)$ with
$E\subseteq V\times V$, the graph algebra $\mathrm{Alg}(D)=V\cup\{0\}$ with
product
\[
xy=\begin{cases}
x, & (x,y)\in E,\\
0, & \text{otherwise}
\end{cases}
\]
as the state set: the automaton $\mathrm{GAtm}(D,T)$ has states
$\mathrm{Alg}(D)^{1}$, initial state $1$, terminal set $T$, and next-state map
$a\cdot x=f(x)a$ for a labelling $f\colon \Gamma\to \mathrm{Alg}(D)^{1}$. (We
write $\mathrm{GAtm}$ to avoid confusion with the universal automaton
$\mathrm{Atm}(A,B)$ of Section~\ref{sec:intro}.) Languages recognised by such
automata form a proper subclass of the regular languages, closed under Kleene
star and complement.

Homomorphisms of hypergraphs have also been studied from the homotopy point of
view: two homomorphisms are weakly homotopic iff they are path-connected in
$H^{H'}$, and homotopic in the stronger sense iff their induced poset maps are
topologically homotopic~\cite{zhang2023}.
\subsection{Epimorphisms of semigroups}

A morphism $\alpha\colon S\to T$ is an epimorphism iff for all $\beta,\gamma$
with $\alpha\beta=\alpha\gamma$ we have $\beta=\gamma$~\cite{almeida2025}. In a
concrete category every onto (surjective) morphism is an epimorphism, but the
converse fails for semigroups~\cite{almeida2025}. The obstruction is captured
by Isbell's dominion
\[
\mathrm{Dom}(U,S)=\{d\in S:\ \forall\alpha,\beta\colon S\to T,\
\alpha|_U=\beta|_U\Rightarrow d\alpha=d\beta\}
\]
and by Isbell's Zigzag Theorem: $d\in\mathrm{Dom}(U,S)$ iff $d\in U$ or there
are $u_0,\dots,u_{2m}\in U$ and $x_1,\dots,x_m,y_1,\dots,y_m\in S$ with
\begin{align*}
&d=u_0y_1,\qquad u_0=x_1u_1,\\
&u_{2i-1}y_i=u_{2i}y_{i+1},\qquad x_iu_{2i}=x_{i+1}u_{2i+1}\qquad(1\le i\le m-1),\\
&u_{2m-1}y_m=u_{2m},\qquad x_mu_{2m}=d,
\end{align*}
see~\cite{almeida2025}.
A subsemigroup $U$ is epimorphically embedded if $\mathrm{Dom}(U,S)=S$, and
saturated if it cannot be properly epimorphically embedded~\cite{almeida2025}.
The decisive modern results classify the pseudovarieties $V$ for which
epimorphisms into members are automatically surjective~\cite{almeida2025}.
Every epimorphism to a member of $V$ is onto iff $V$ excludes the $2\times 2$
Brandt aperiodic semigroup $B_2$; the only finite obstruction to
epi-surjectivity from members is the semigroup $Y=B_2\setminus\{b\}$; precisely
three maximal pseudovarieties avoid $Y$, they are $F$-saturated, and $S$ (all
finite semigroups) is the only $F$-epimorphically closed pseudovariety
containing $Y$~\cite{almeida2025}.

For a stable semigroup $S$ and a finite subsemigroup $U$ epimorphically
embedded in $S$, every $J$-maximal element of $S\setminus U$ is regular, and
one can find zigzags whose factors all lie in the $D$-class of that
element~\cite{almeida2025}. This connects epimorphism theory to the internal
structure of the input symbol semigroup, because the Cayley representation
realises any finite semigroup as the transition semigroup of a complete
semiautomaton $(S^{1},\Gamma,\delta)$ with $\delta(s,a)=sa$~\cite{almeida2025}.

\subsection{Torsion theory and automaton semigroups}

Lam shows a bijective correspondence between hereditary pretorsion classes
(HPC) of $S$-automata, the right linear topologies on the input semigroup
$S^{1}$, and idempotent preradicals $\tau$ on the category of $S$-automata
with $C_\tau=\{M:\tau(M)=M\}$ closed under subautomata and finite
products~\cite{lam1993}. The correspondence uses the fact that
$S^{1}/\ker(\alpha_x)\cong xS^{1}$, so that the topology on the input
semigroup encodes exactly the quotients allowed~\cite{lam1993}. This gives a
torsion-theoretic reading of which input-semigroup quotients (i.e.\ which
automaton epimorphisms) preserve membership in a class.

Cain's theory of automaton semigroups takes the input symbol semigroup as
intrinsic: an automaton is a deterministic real-time transducer $(Q,B,\delta)$,
and the generated automaton semigroup $\Sigma(A)$ is the subsemigroup of the
tree-endomorphism semigroup $\mathrm{End}\,B^{*}$ generated by the
states~\cite{cain2009}. The partial-automaton extension sharpens this:
D'Angeli, Rodaro and W\"achter show that a semigroup $S$ is a (partial)
automaton semigroup iff its structure is compatible with removing zeros from a
complete automaton, that $S^{0}$ is always a complete automaton semigroup, and
that no direct or semidirect product of an arbitrary semigroup with a
non-trivial subsemigroup of the free monogenic semigroup is an automaton
semigroup~\cite{dangeli2020}. They also prove that automaton-inverse
semigroups coincide with inverse automaton semigroups~\cite{dangeli2020}.

The results of~\cite{almeida2025} concern epimorphisms in the categorical
sense, which need not be surjective, whereas the works~\cite{yang1975,bavel1968}
treat homomorphisms of automata without hypergraph structure on the state and
output sets. In this paper, an
epimorphism is a surjective homomorphism, and the objects are the universal
hypergraphic automata $\mathrm{Atm}(H_X,H_Y)$ with input symbol semigroups
$S(H_X,H_Y)$. Isomorphisms~\cite{khvorostukhina2017} and
monomorphisms~\cite{khvorostukhina2025} of such automata have been described
before; the description of epimorphisms is the subject of
Sections~\ref{sec:main} and~\ref{sec:applications}.

\section{Preliminaries}\label{sec001prelim}

From now on mappings are written on the left and composed as functions:
$(f\circ\varphi)(x)=f(\varphi(x))$. For a mapping $\varphi$ and a set $Z$ we put
$\varphi(Z)=\{\varphi(z):z\in Z\}$. A constant mapping with value $a$ is denoted
$C_a$, and the identity transformation of $X$ is $\Delta_X$. The notation $f^{-1}$
for a mapping is used \emph{only} when $f$ is bijective; preimages of sets are
written $f^{-1}(Z)$ only in the set-theoretic sense and are never used as
mappings unless $f$ is a bijection. Here $\mathbb N=\{1,2,\dots\}$.

\begin{definition}\label{def0hypergraph}
A \emph{hypergraph} is a pair $H=(X,L)$, where $X\neq\varnothing$ is the set of
vertices and $L$ is a family of subsets of $X$ called edges. A set
$Z\subseteq X$ is \emph{bounded} in $H$ if $Z\subseteq l$ for some $l\in L$, and
\emph{unbounded} otherwise. The hypergraph $H$ is \emph{effective} if every
vertex belongs to at least one edge.
\end{definition}

\begin{definition}\label{def0pdef}
Let $p\in\mathbb N$. A hypergraph $H=(X,L)$ has \emph{$p$-definable edges} if
$|l|\ge p+1$ for every $l\in L$, and any $p$ distinct vertices lie in at most
one common edge; equivalently, $|l\cap l'|<p$ for all distinct $l,l'\in L$.
\end{definition}

\begin{definition}\label{def0phyper}
Let $p\in\mathbb N$. A hypergraph $H=(X,L)$ is a \emph{$p$-hypergraph} if
\begin{enumerate}
\item[(A1)] any $p$ distinct vertices belong to exactly one common edge;
\item[(A2)] every edge contains at least $p+1$ vertices;
\item[(A3)] some $(p+1)$-element subset of $X$ is not bounded in $H$.
\end{enumerate}
\end{definition}
Every $p$-hypergraph is effective with $p$-definable edges. Indeed, by (A3)
$|X|\ge p+1$, so every vertex lies in some $p$-element subset, hence in an edge
by (A1); and two distinct edges sharing $p$ vertices would contradict the
uniqueness in (A1). Projective planes, and affine planes of order $q\ge 3$, with
points as vertices and lines as edges, are $2$-hypergraphs
\cite{karteszi,hartshorne}. (The affine plane of order~$2$ has lines with two
points and violates (A2).)

Let $H_X=(X,L_X)$ and $H_Y=(Y,L_Y)$ be hypergraphs.

\begin{definition}\label{def00hom}
A mapping $\varphi\colon X\to Y$ is a \emph{homomorphism} $H_X\to H_Y$ if
$\varphi(l)$ is bounded in $H_Y$ for every $l\in L_X$, i.e.\
$(\forall l\in L_X)(\exists l'\in L_Y)\ \varphi(l)\subseteq l'$. The set of all
homomorphisms is $\mathrm{Hom}(H_X,H_Y)$, and $\mathrm{End}\,H_X=\mathrm{Hom}(H_X,H_X)$.
A homomorphism is a \emph{monomorphism} if it is injective. A bijection
$\varphi\colon X\to Y$ is an \emph{isomorphism} if $l\in L_X\Leftrightarrow
\varphi(l)\in L_Y$.
\end{definition}

Compositions of homomorphisms are homomorphisms, so $\mathrm{End}\,H_X$ is a
semigroup under composition. If $H_Y$ is effective, then every mapping
$X\to l$ with $l\in L_Y$ is a homomorphism; in particular all constant mappings
$C_b$, $b\in Y$, belong to $\mathrm{Hom}(H_X,H_Y)$.

The following lemma shows that, for hypergraphs with $p$-definable edges, a
bijection which is a homomorphism in both directions is already an isomorphism.

\begin{lemma}\label{lem:iso}
Let $H_X,H_Y$ have $p$-definable edges and let $\varphi\colon X\to Y$ be a
bijection such that both $\varphi$ and $\varphi^{-1}$ are homomorphisms. Then
$\varphi$ is an isomorphism.
\end{lemma}

\begin{proof}
Let $l\in L_X$. Choose $l'\in L_Y$ with $\varphi(l)\subseteq l'$, and then
$l''\in L_X$ with $\varphi^{-1}(l')\subseteq l''$. Then $l\subseteq l''$, so
$|l\cap l''|=|l|\ge p+1$, and $p$-definability gives $l=l''$. Hence
$l\subseteq\varphi^{-1}(l')\subseteq l$, i.e.\ $\varphi(l)=l'\in L_Y$.
Conversely, let $l'\in L_Y$. Choose $l''\in L_X$ with
$\varphi^{-1}(l')\subseteq l''$ and $l'''\in L_Y$ with
$\varphi(l'')\subseteq l'''$. Then $l'\subseteq l'''$, so $l'=l'''$ by
$p$-definability of $H_Y$, and $\varphi(l'')=l'$. Thus $l'$ is the image of an
edge of $H_X$.
Finally, if $\varphi(l)\in L_Y$ for some $l\subseteq X$, then $\varphi(l)=\varphi(l'')$
for some $l''\in L_X$ by the previous paragraph, so $l=l''\in L_X$ by injectivity.
\end{proof}

\subsection{Weak and strong epimorphisms}

\begin{definition}\label{def:weak-strong}
A homomorphism $g\colon H_Y\to H_{Y_1}$ is a \emph{weak epimorphism} if
$g(Y)=Y_1$, and a \emph{strong epimorphism} if, in addition,
$(\forall l_1\in L_{Y_1})(\exists l\in L_Y)\ g(l)=l_1$.
\end{definition}

Every strong epimorphism is weak. As stated in Section~\ref{sec:intro}, ``epimorphism'' always
means \emph{surjective} homomorphism.

\begin{definition}[Image hypergraph]\label{def:image}
For a homomorphism $g\colon H_Y\to H_{Y_1}$ the \emph{image hypergraph} is
$g(H_Y)=(g(Y),\{g(l):l\in L_Y\})$.
\end{definition}

Even if $H_Y$ and $H_{Y_1}$ have $p$-definable edges, $g(H_Y)$ need not: $g$ may
collapse an edge to fewer than $p+1$ vertices, or the images of two edges may
share $p$ or more vertices. Hence $g(H_Y)$ need not belong to the class studied
in this paper.

\begin{definition}[\cite{plotkin}]\label{def0automaton}
An \emph{automaton} is a system
$\mathcal A=(X,S,Y,\delta,\lambda)$ with $S$ a semigroup and
$\delta\colon X\times S\to X$, $\lambda\colon X\times S\to Y$ such that
\[
\delta(x,ab)=\delta(\delta(x,a),b),\qquad
\lambda(x,ab)=\lambda(\delta(x,a),b)\qquad(x\in X,\ a,b\in S).
\]
Put $\delta_s=\delta(\cdot,s)$ and $\lambda_s=\lambda(\cdot,s)$. It is a
\emph{hypergraphic automaton} if $X$ and $Y$ carry hypergraph structures
$H_X,H_Y$ and $\delta_s\in\mathrm{End}\,H_X$, $\lambda_s\in\mathrm{Hom}(H_X,H_Y)$
for all $s\in S$.
\end{definition}

\begin{definition}\label{def0universal}
For hypergraphs $H_X,H_Y$ let
$S(H_X,H_Y)=\mathrm{End}\,H_X\times\mathrm{Hom}(H_X,H_Y)$ with multiplication
\[
(\varphi,\psi)(\varphi_1,\psi_1)=(\varphi_1\circ\varphi,\ \psi_1\circ\varphi).
\]
This is a semigroup, and it is exactly the product
$(\varphi\varphi_1,\varphi\psi_1)$ of~\cite{plotkin} written in right notation.
The \emph{universal hypergraphic automaton} is
$\mathrm{Atm}(H_X,H_Y)=(H_X,S(H_X,H_Y),H_Y,\delta',\lambda')$ with
$\delta'(x,(\varphi,\psi))=\varphi(x)$ and $\lambda'(x,(\varphi,\psi))=\psi(x)$.
For every hypergraphic automaton $\mathcal A=(H_X,S,H_Y,\delta,\lambda)$ the map
$s\mapsto(\delta_s,\lambda_s)$ is a semigroup homomorphism $S\to S(H_X,H_Y)$
\cite{plotkin}.
\end{definition}

\subsection{Morphisms of universal automata and the induced map}

\begin{definition}\label{def0morphism}
Let $f\colon X\to X_1$ and $g\colon Y\to Y_1$ be mappings and
$\ppi\colon S(H_X,H_Y)\to S(H_{X_1},H_{Y_1})$ a semigroup homomorphism. Denote
by $\delta_1',\lambda_1'$ the transition and output functions of
$\mathrm{Atm}(H_{X_1},H_{Y_1})$. The triple $\gamma=(f,\ppi,g)$ is
\emph{compatible} if for all $x\in X$ and $s=(\varphi,\psi)$ with
$\ppi(s)=(\chi,\omega)$
\[
f(\delta'(x,s))=\delta_1'(f(x),\ppi(s)),\qquad
g(\lambda'(x,s))=\lambda_1'(f(x),\ppi(s)),
\]
that is,
\begin{equation}\label{eq:compat}
f\circ\varphi=\chi\circ f,\qquad g\circ\psi=\omega\circ f.
\end{equation}
A compatible triple is a \emph{morphism} $\mathrm{Atm}(H_X,H_Y)\to
\mathrm{Atm}(H_{X_1},H_{Y_1})$ if $f\in\mathrm{Hom}(H_X,H_{X_1})$ and
$g\in\mathrm{Hom}(H_Y,H_{Y_1})$. A morphism is an \emph{epimorphism} if $f$,
$\ppi$, $g$ are all surjective, and a \emph{strong epimorphism} if moreover $g$
is a strong epimorphism of hypergraphs.
\end{definition}

\begin{lemma}\label{lem0unique}
If $f$ is surjective, then $\ppi$ is uniquely determined by $(f,g)$:
if $(f,\ppi,g)$ and $(f,\ppi',g)$ are compatible, then $\ppi=\ppi'$. If $f$ is
bijective, then necessarily
\begin{equation}\label{eq:induced}
\ppi(\varphi,\psi)=\bigl(f\circ\varphi\circ f^{-1},\ g\circ\psi\circ f^{-1}\bigr).
\end{equation}
\end{lemma}
\begin{proof}
Let $\ppi(\varphi,\psi)=(\chi,\omega)$ and $\ppi'(\varphi,\psi)=(\chi',\omega')$.
By \eqref{eq:compat}, $\chi\circ f=f\circ\varphi=\chi'\circ f$ and
$\omega\circ f=g\circ\psi=\omega'\circ f$. A surjective $f$ is
right-cancellable, so $\chi=\chi'$ and $\omega=\omega'$. For bijective $f$
compose with $f^{-1}$ on the right.
\end{proof}
For bijective $f$ we therefore write $\ppi=(f^{2},f\times g)$, where
$f^{2}(\varphi)=f\circ\varphi\circ f^{-1}$ and
$(f\times g)(\psi)=g\circ\psi\circ f^{-1}$. For $f$ merely surjective
and not injective, a compatible $\ppi$ defined on all of $S(H_X,H_Y)$ exists only
if
\begin{equation}\label{eq:kernel}
\ker f\subseteq\ker(f\circ\varphi)\ \text{ and }\ \ker f\subseteq\ker(g\circ\psi)
\quad\text{for all }(\varphi,\psi)\in S(H_X,H_Y),
\end{equation}
because $\chi\circ f$ and $\omega\circ f$ must factor through $f$. In
Section~\ref{sec:main} we show that, under hypotheses (V) and (P) of
Theorem~\ref{thm0main}, the existence of a surjective compatible $\ppi$ (which
implies \eqref{eq:kernel}) forces $f$ to be an isomorphism;
Example~\ref{ex0V1needed} shows that this fails without (V).
\begin{lemma}\label{lem0induced}
Let $f\colon H_X\to H_{X_1}$ be an isomorphism and $g\colon H_Y\to H_{Y_1}$ a
homomorphism. Then \eqref{eq:induced} defines a semigroup homomorphism
$\ppi\colon S(H_X,H_Y)\to S(H_{X_1},H_{Y_1})$ and $(f,\ppi,g)$ is compatible.
\end{lemma}

\begin{proof}
Both $f$ and $f^{-1}$ are homomorphisms. Hence $f\circ\varphi\circ f^{-1}$ is an
endomorphism of $H_{X_1}$, and $g\circ\psi\circ f^{-1}$ is a composition of
homomorphisms, so lies in $\mathrm{Hom}(H_{X_1},H_{Y_1})$. For
$s=(\varphi,\psi)$, $s_1=(\varphi_1,\psi_1)$,
\begin{align*}
\ppi(s)\ppi(s_1)
&=\bigl((f\varphi_1f^{-1})\circ(f\varphi f^{-1}),\
(g\psi_1f^{-1})\circ(f\varphi f^{-1})\bigr)\\
&=\bigl(f\varphi_1\varphi f^{-1},\ g\psi_1\varphi f^{-1}\bigr)=\ppi(ss_1).
\end{align*}
Compatibility \eqref{eq:compat} is immediate from \eqref{eq:induced}.
\end{proof}

\subsection{Lifting}

The map $f\times g$ is surjective only if homomorphisms into $H_{Y_1}$ can be
lifted along $g$ to homomorphisms into $H_Y$. This is a genuine condition and
does not follow from surjectivity of $g$, even on edges.

\begin{definition}[Lifting property]\label{def0lifting}
Let $H$ be a hypergraph. A homomorphism $g\colon H_Y\to H_{Y_1}$ is
\emph{$H$-liftable} if the map
$\mathrm{Hom}(H,H_Y)\to\mathrm{Hom}(H,H_{Y_1})$, $\psi\mapsto g\circ\psi$, is
surjective.
\end{definition}

\begin{lemma}\label{lem0lift-equiv}
Let $f\colon H_X\to H_{X_1}$ be an isomorphism and $g\colon H_Y\to H_{Y_1}$ a
homomorphism. The map $\psi\mapsto(f\times g)(\psi)$ from
$\mathrm{Hom}(H_X,H_Y)$ to $\mathrm{Hom}(H_{X_1},H_{Y_1})$ is surjective if and
only if $g$ is $H_X$-liftable.
\end{lemma}

\begin{proof}
Since $f$ is an isomorphism, $\psi_1\mapsto\psi_1\circ f$ is a bijection
$\mathrm{Hom}(H_{X_1},H_{Y_1})\to\mathrm{Hom}(H_X,H_{Y_1})$ (its inverse is
$\chi\mapsto\chi\circ f^{-1}$), and
$(f\times g)(\psi)=\psi_1$ if and only if $g\circ\psi=\psi_1\circ f$.
Hence every $\psi_1\in\mathrm{Hom}(H_{X_1},H_{Y_1})$ is of the form
$(f\times g)(\psi)$ if and only if every $\chi\in\mathrm{Hom}(H_X,H_{Y_1})$ is of
the form $g\circ\psi$ with $\psi\in\mathrm{Hom}(H_X,H_Y)$, which is
$H_X$-liftability of $g$.
\end{proof}

\begin{lemma}\label{lem0lift1weak}
Let $g\colon H_Y\to H_{Y_1}$ be a homomorphism and let $H_{Y_1}$ be effective.
If $g$ is $H$-liftable for some hypergraph $H$, then $g$ is a weak epimorphism.
\end{lemma}

\begin{proof}
For $b\in Y_1$ the constant map $C_b$ is in $\mathrm{Hom}(H,H_{Y_1})$, since
$H_{Y_1}$ is effective. A lift $\psi$ satisfies $g\circ\psi=C_b$; taking any
vertex $x$ of $H$ gives $b=g(\psi(x))\in g(Y)$.
\end{proof}

The converse fails, even for strong epimorphisms.

\begin{example}\label{ex0notliftable}
Let $p=2$. Let $H_Y$ have two disjoint edges $A=\{a_1,a_2,a_3\}$ and
$B=\{b_1,b_2,b_3\}$. Let $H_{Y_1}$ have vertices $\{a_1,a_2,y,b_1,b_2\}$ and
edges $l_1=\{a_1,a_2,y\}$, $l_2=\{y,b_1,b_2\}$ (they meet in one vertex). Let
$g$ fix $a_1,a_2,b_1,b_2$ and send $a_3,b_3\mapsto y$. Then $g(A)=l_1$ and
$g(B)=l_2$, so $g$ is a strong epimorphism. Let $H_X$ have vertices
$x,u_1,u_2,w_1,w_2$ and edges $e=\{x,u_1,u_2\}$, $e'=\{x,w_1,w_2\}$. All three
hypergraphs are effective with $2$-definable edges. Define
$\psi_1\colon x\mapsto y$, $u_i\mapsto a_i$, $w_i\mapsto b_i$; then
$\psi_1(e)=l_1$ and $\psi_1(e')=l_2$, so $\psi_1\in\mathrm{Hom}(H_X,H_{Y_1})$.
Suppose $g\circ\psi=\psi_1$ for some $\psi\in\mathrm{Hom}(H_X,H_Y)$. Since
$\psi(e)$ and $\psi(e')$ are each inside $A$ or $B$ and share the vertex
$\psi(x)$, both lie in the same edge, so $g\psi(X)\subseteq l_1$ or
$\subseteq l_2$. But $\psi_1(X)\not\subseteq l_1$ and $\psi_1(X)\not\subseteq l_2$,
a contradiction. Hence the strong epimorphism $g$ is not $H_X$-liftable.
\end{example}

Fix $x_0\in X$ and a subset $F\subseteq Y$ (possibly empty). The language
\[
L(x_0,F)=\{(\varphi,\psi)\in S(H_X,H_Y):\ \psi(x_0)\in F\}\subseteq S(H_X,H_Y)
\]
is the set of input symbols that $\mathrm{Atm}(H_X,H_Y)$, started in $x_0$,
maps to an output symbol in $F$. A subset of $S(H_X,H_Y)$ is
\emph{recognisable} if it equals $L(x_0,F)$ for some $x_0$ and $F$. For a
\emph{fixed} $x_0$ the family $\{L(x_0,F)\}_{F\subseteq Y}$ is closed under unions,
intersections and complements: $L(x_0,F)\cup L(x_0,F')=L(x_0,F\cup F')$,
$L(x_0,F)\cap L(x_0,F')=L(x_0,F\cap F')$ and
$S(H_X,H_Y)\setminus L(x_0,F)=L(x_0,Y\setminus F)$.
For different initial states this closure is not claimed.

\subsection{Problem statement}\label{ssec0problem}

Let $\mathrm{Atm}(H_X,H_Y)$ and $\mathrm{Atm}(H_{X_1},H_{Y_1})$ be universal
automata over effective hypergraphs with $p$-definable edges, and let
$f\colon X\to X_1$ be a surjection and $g\colon Y\to Y_1$ a mapping. We ask for
necessary and sufficient conditions on $f$ and $g$ under which there is a
semigroup homomorphism $\ppi$ such that $(f,\ppi,g)$ is a compatible triple
with $\ppi$ surjective (such $\ppi$ is unique by Lemma~\ref{lem0unique}, and is
given by \eqref{eq:induced} once $f$ is known to be bijective), and the
corresponding conditions for $(f,\ppi,g)$ to be an epimorphism of automata.

\section{Main results}\label{sec:main}

Throughout this section $H_X,H_Y,H_{X_1},H_{Y_1}$ are effective hypergraphs with
$p$-definable edges, $f\colon X\to X_1$ is a surjection and $g\colon Y\to Y_1$ is
a mapping. A triple $(f,\ppi,g)$ is \emph{compatible} in the sense of
Definition~\ref{def0morphism}. For a bijective $f$, Lemma~\ref{lem0unique} gives
$\ppi(\varphi,\psi)=(f\varphi f^{-1},\,g\psi f^{-1})$.

\begin{lemma}\label{lem0small-unbounded}
Let $H=(Z_0,L)$ have $p$-definable edges and let $Z\subseteq Z_0$ be unbounded
with $|Z|\ge p+1$. Then $Z$ contains an unbounded subset with exactly $p+1$
elements. Consequently, every unbounded set contains a finite unbounded subset
with at most $p+1$ elements.
\end{lemma}
\begin{proof}
Suppose every $(p+1)$-subset of $Z$ is bounded. Fix a $p$-subset $T\subseteq Z$.
Since $|Z|\ge p+1$, there is $z\in Z\setminus T$, and $T\cup\{z\}$ lies in an
edge; by $p$-definability the edge $r_T$ containing $T$ is unique. Every
$z\in Z\setminus T$ lies in an edge containing $T$, because $T\cup\{z\}$ is a
$(p+1)$-subset of $Z$, hence in $r_T$. So $Z\subseteq r_T$, a contradiction.
For the last statement: if $|Z|\le p$ take $Z$ itself.
\end{proof}

\begin{lemma}\label{lem0small-sets}
Let $V=\{v_1,\dots,v_{p+1}\}$ have all $p$-element subsets bounded in $H_X$. Then
every subset of $V$ with at most $p$ elements is bounded in $H_X$.
\end{lemma}
\begin{proof}
A subset $S\subseteq V$ with $|S|\le p$ is contained in a $p$-element subset $T$
of $V$, and $T\subseteq l$ for some $l\in L_X$, so $S\subseteq l$.
\end{proof}

\begin{lemma}\label{lem0V1auto}
Suppose $H_X$ satisfies (A1) and $H_{X_1}$ satisfies (A3). Then every surjection
$f\colon X\to X_1$ admits a set $V$ as in hypothesis~(V) of
Theorem~\ref{thm0main}.
\end{lemma}
\begin{proof}
By (A3) there is an unbounded $W\subseteq X_1$ with $|W|=p+1$. Choose
$V\subseteq X$ with $f(V)=W$ and $|V|=p+1$ (possible as $f$ is onto: take one
preimage of each element of $W$). By (A1) in $H_X$, every $p$-subset of $V$ is
bounded, while $f(V)=W$ is unbounded.
\end{proof}

\subsection{The main theorem}

\begin{theorem}\label{thm0main}
Let $f\colon X\to X_1$ be a surjection and $g\colon Y\to Y_1$ a mapping.
Assume:
\begin{itemize}
\item[(V)] there is $V=\{v_1,\dots,v_{p+1}\}\subseteq X$ such that every
$p$-element subset of $V$ is bounded in $H_X$, but $f(V)$ is unbounded in
$H_{X_1}$;
\item[(P)] $p\ge 2$ (for $p=1$ one must assume in addition that $f$ is injective).
\end{itemize}
Then the following are equivalent:
\begin{enumerate}
\item there is a surjective semigroup homomorphism
$\ppi\colon S(H_X,H_Y)\to S(H_{X_1},H_{Y_1})$ such that $(f,\ppi,g)$ is compatible;
\item $f$ is an isomorphism $H_X\to H_{X_1}$, and $g$ is a homomorphism which is
$H_X$-liftable (Definition~\ref{def0lifting});
\item there is $\ppi$ such that $\gamma=(f,\ppi,g)$ is an epimorphism of automata
$\mathrm{Atm}(H_X,H_Y)\to\mathrm{Atm}(H_{X_1},H_{Y_1})$.
\end{enumerate}
In this case $g$ is a weak epimorphism and $\ppi=(f^2,f\times g)$ is given by
\eqref{eq:induced}.
\end{theorem}

\begin{proof}
$(3)\Rightarrow(1)$ is immediate from the definitions.

$(2)\Rightarrow(1)$. By Lemma~\ref{lem0induced}, \eqref{eq:induced} defines a
semigroup homomorphism and $(f,\ppi,g)$ is compatible. Surjectivity: let
$(\chi,\omega)\in S(H_{X_1},H_{Y_1})$. Put $\varphi=f^{-1}\chi f$; it is an
endomorphism of $H_X$ since $f,f^{-1},\chi$ are homomorphisms. The map
$\omega f$ lies in $\mathrm{Hom}(H_X,H_{Y_1})$, so by $H_X$-liftability there is
$\psi\in\mathrm{Hom}(H_X,H_Y)$ with $g\psi=\omega f$. Then
$\ppi(\varphi,\psi)=(f\varphi f^{-1},g\psi f^{-1})=(\chi,\omega)$.

$(1)\Rightarrow(2)$. Let $\ppi$ be as in (1). For $(\varphi,\psi)\in S(H_X,H_Y)$
write $\ppi(\varphi,\psi)=(\chi,\omega)$; compatibility reads
\begin{equation}\label{eq:comp2}
f\varphi=\chi f,\qquad g\psi=\omega f.
\end{equation}
Recall that all constant maps are homomorphisms, since the hypergraphs are
effective. The proof proceeds in six steps: Steps~1--4 show that $f$ is an
isomorphism, Step~5 that $g$ is a homomorphism, and Step~6 that $g$ is
$H_X$-liftable.

\medskip
\noindent\textbf{Step~1: $f$ is injective.}
If $p=1$ this is assumed, so let $p\ge2$.
Suppose $f(a)=f(b)$ with $a\neq b$. Since $f(V)$ is unbounded and $H_{X_1}$ is
effective, $f$ is not constant on $V$; pick $v_i,v_j\in V$ with
$f(v_i)\ne f(v_j)$. Define $\varphi\colon X\to X$ by $\varphi(a)=v_i$ and
$\varphi(x)=v_j$ for $x\neq a$. Then $\varphi(X)\subseteq\{v_i,v_j\}$, which is
bounded by Lemma~\ref{lem0small-sets} (here $p\ge2$ is used), so
$\varphi\in\mathrm{End}\,H_X$. Apply \eqref{eq:comp2} to $(\varphi,C_d)$, with
$d\in Y$ arbitrary: $f\varphi=\chi f$, hence
$f\varphi(a)=\chi(f(a))=\chi(f(b))=f\varphi(b)$, i.e.\
$f(v_i)=f(v_j)$, a contradiction. Thus $f$ is a bijection and, by
Lemma~\ref{lem0unique}, $\ppi$ is given by \eqref{eq:induced}; in particular
\begin{equation}\label{eq:end}
f\varphi f^{-1}\in\mathrm{End}\,H_{X_1}\quad\text{for all }\varphi\in\mathrm{End}\,H_X.
\end{equation}

\medskip
\noindent\textbf{Step~2: $f$ is a homomorphism.}
Suppose $l\in L_X$ and $f(l)$ is
unbounded in $H_{X_1}$. As $|f(l)|=|l|\ge p+1$, Lemma~\ref{lem0small-unbounded}
gives an unbounded $W_1\subseteq f(l)$ with $|W_1|=p+1$. Fix an edge
$e_1\in L_{X_1}$ and a map $\chi\colon X_1\to W_1$ with $\chi(e_1)=W_1$ (possible,
since $|e_1|\ge p+1$). Put $\varphi=f^{-1}\chi f$. Then
$\varphi(X)\subseteq f^{-1}(W_1)\subseteq l$, so $\varphi\in\mathrm{End}\,H_X$,
whereas $f\varphi f^{-1}=\chi$ maps $e_1$ onto the unbounded set $W_1$. This
contradicts \eqref{eq:end}.

\medskip
\noindent\textbf{Step~3: $f^{-1}$ is a homomorphism.}
Suppose $e_1\in L_{X_1}$ and
$Y_0:=f^{-1}(e_1)$ is unbounded in $H_X$; then $|Y_0|\ge p+1$. We construct
$\varphi\in\mathrm{End}\,H_X$ and distinct $y_1,\dots,y_{p+1}\in Y_0$ with
$\varphi(y_i)=v_i$. Then $f\varphi f^{-1}(e_1)\supseteq f(V)$ is unbounded,
contradicting \eqref{eq:end}.

\smallskip
\noindent\emph{Subcase A: every $p$-subset of $Y_0$ is unbounded in $H_X$.} Then
$|r\cap Y_0|\le p-1$ for every edge $r$. Take any distinct
$y_1,\dots,y_{p+1}\in Y_0$ and put $\varphi(y_i)=v_i$ for $i\le p$ and
$\varphi(x)=v_{p+1}$ otherwise. For an edge $r$,
$\varphi(r)\subseteq\{v_i:y_i\in r\}\cup\{v_{p+1}\}$ has at most $p$ elements,
hence is bounded by Lemma~\ref{lem0small-sets}.

\smallskip
\noindent\emph{Subcase B: some $p$-subset $\{y_1,\dots,y_p\}\subseteq Y_0$ lies in an edge
$r_0$.} Since $Y_0$ is unbounded, there is $y_{p+1}\in Y_0\setminus r_0$. Define
\[
\varphi(x)=\begin{cases}
v_i & x=y_i,\ i\le p-1,\\
v_p & x\in r_0\setminus\{y_1,\dots,y_{p-1}\},\\
v_{p+1} & x\notin r_0 .
\end{cases}
\]
Then $\varphi(y_p)=v_p$ and $\varphi(y_{p+1})=v_{p+1}$. Also
$\varphi(r_0)\subseteq\{v_1,\dots,v_p\}$ is bounded. For an edge $r'\neq r_0$
we have $|r'\cap r_0|\le p-1$ by $p$-definability, so $\varphi(r')$ consists of
at most $p-1$ values from $\{v_1,\dots,v_p\}$ together with possibly $v_{p+1}$:
at most $p$ elements, bounded by Lemma~\ref{lem0small-sets}. Hence
$\varphi\in\mathrm{End}\,H_X$.

\medskip
\noindent\textbf{Step~4: $f$ is an isomorphism.}
By Steps~1--3, $f$ is a bijection and
$f,f^{-1}$ are homomorphisms; apply Lemma~\ref{lem:iso}.

\medskip
\noindent\textbf{Step~5: $g$ is a homomorphism.}
Suppose $m\in L_Y$ and $g(m)$ is
unbounded. By Lemma~\ref{lem0small-unbounded} choose a finite unbounded
$T\subseteq g(m)$ with $|T|\le p+1$, and $u_t\in m$ with $g(u_t)=t$ for $t\in T$.
Fix $r\in L_X$ and let $\psi\colon X\to m$ map $r$ onto $\{u_t:t\in T\}$ (note
$|r|\ge p+1\ge|T|$). Then $\psi\in\mathrm{Hom}(H_X,H_Y)$, and by
\eqref{eq:comp2} applied to $(\Delta_X,\psi)$ we get
$\omega=g\psi f^{-1}\in\mathrm{Hom}(H_{X_1},H_{Y_1})$. But
$f(r)\in L_{X_1}$ and $\omega(f(r))=g\psi(r)=T$ is unbounded, a contradiction.

\medskip
\noindent\textbf{Step~6: $g$ is $H_X$-liftable.}
Let $\psi_1\in\mathrm{Hom}(H_X,H_{Y_1})$.
Then $\omega_1=\psi_1f^{-1}\in\mathrm{Hom}(H_{X_1},H_{Y_1})$, and
$(\Delta_{X_1},\omega_1)\in S(H_{X_1},H_{Y_1})$. By surjectivity of $\ppi$ there
is $(\varphi,\psi)$ with $\ppi(\varphi,\psi)=(\Delta_{X_1},\omega_1)$, i.e.\
$g\psi f^{-1}=\omega_1$, i.e.\ $g\psi=\psi_1$.

\medskip
$(2)\Rightarrow(3)$. By $(2)\Rightarrow(1)$ we have a surjective $\ppi$ with
$(f,\ppi,g)$ compatible; $f$ and $g$ are homomorphisms, $f$ is bijective, and
$g$ is surjective by Lemma~\ref{lem0lift1weak}. The last statement of the
theorem follows from Lemma~\ref{lem0lift1weak} and Lemma~\ref{lem0unique}.
\end{proof}

\begin{example}\label{ex0V1needed}
Let $p\ge1$ and $m\ge p+1$. Let $H_X$ consist of $m$ pairwise disjoint edges
$r_1,\dots,r_m$ of size $p+1$, and let $H_{X_1}$ have $X_1=\{1,\dots,m\}$ and the
single edge $X_1$. Let $f(x)=k$ for $x\in r_k$. Similarly let $H_Y$ consist of
$m$ disjoint edges $s_1,\dots,s_m$ of size $p+1$, $H_{Y_1}$ have vertex set
$\{1,\dots,m\}$ and the single edge $Y_1$, and let $g(y)=k$ for $y\in s_k$. All
four hypergraphs are effective with $p$-definable edges. Since $X_1$ and $Y_1$
are themselves edges, \emph{every} mapping $X_1\to X_1$ is an endomorphism and
every mapping $X_1\to Y_1$ is a homomorphism.
\begin{itemize}
\item For $(\varphi,\psi)\in S(H_X,H_Y)$ we have $\varphi(r_k)\subseteq r_{k'}$
and $\psi(r_k)\subseteq s_{k''}$ for some $k',k''$, so $f\varphi$ and $g\psi$ are
constant on each $r_k$. Hence $f\varphi=\chi f$ and $g\psi=\omega f$ for unique
$\chi,\omega$ ($f$ is onto), and $(\chi,\omega)\in S(H_{X_1},H_{Y_1})$. Put
$\ppi(\varphi,\psi)=(\chi,\omega)$. If $(\varphi_1,\psi_1)\mapsto(\chi_1,\omega_1)$,
then $f\varphi_1\varphi=\chi_1f\varphi=\chi_1\chi f$ and
$g\psi_1\varphi=\omega_1f\varphi=\omega_1\chi f$, so by uniqueness
$\ppi\bigl((\varphi,\psi)(\varphi_1,\psi_1)\bigr)=(\chi_1\chi,\omega_1\chi)=
\ppi(\varphi,\psi)\ppi(\varphi_1,\psi_1)$; thus $\ppi$ is a semigroup homomorphism.
\item $\ppi$ is surjective: for $(\chi,\omega)$ let $\varphi$ be constant on each
$r_k$ with value in $f^{-1}(\chi(k))$, and $\psi$ constant on $r_k$ with value in
$g^{-1}(\omega(k))$. Then $\varphi\in\mathrm{End}\,H_X$,
$\psi\in\mathrm{Hom}(H_X,H_Y)$ and $\ppi(\varphi,\psi)=(\chi,\omega)$.
\end{itemize}
Moreover $f$ and $g$ are surjective homomorphisms. Thus $(f,\ppi,g)$ is an
epimorphism of automata although $f$ is not injective. Condition (P) holds when
$p\ge2$, and (V) fails, because every subset of $X_1$ is bounded in $H_{X_1}$.
\end{example}

By Lemma~\ref{lem0lift1weak} and Example~\ref{ex0notliftable}, $H_X$-liftability
is strictly stronger than being a weak, and even a strong, epimorphism. The
following gives a simple sufficient condition.

\begin{proposition}\label{prop:section}
If $g\colon H_Y\to H_{Y_1}$ has a homomorphic section $s\colon H_{Y_1}\to H_Y$
(i.e.\ $s\in\mathrm{Hom}(H_{Y_1},H_Y)$ and $g\circ s=\Delta_{Y_1}$), then $g$ is
$H$-liftable for every hypergraph $H$.
\end{proposition}
\begin{proof}
Let $H$ be a hypergraph and $\psi_1\in\mathrm{Hom}(H,H_{Y_1})$, and put
$\psi=s\circ\psi_1\colon X(H)\to Y$. We check that $\psi$ is a homomorphism and
that it lifts $\psi_1$.

Let $e\in L_H$. Since $\psi_1$ is a homomorphism, there is $l_1\in L_{Y_1}$ with
$\psi_1(e)\subseteq l_1$. Since $s$ is a homomorphism, there is $l\in L_Y$ with
$s(l_1)\subseteq l$. Hence $\psi(e)=s(\psi_1(e))\subseteq s(l_1)\subseteq l$, so
$\psi(e)$ is bounded in $H_Y$, i.e.\ $\psi\in\mathrm{Hom}(H,H_Y)$. By
associativity of composition,
$g\circ\psi=(g\circ s)\circ\psi_1=\Delta_{Y_1}\circ\psi_1=\psi_1$.

Since $\psi_1$ was arbitrary, $\psi\mapsto g\circ\psi$ is a surjection
$\mathrm{Hom}(H,H_Y)\to\mathrm{Hom}(H,H_{Y_1})$, i.e.\ $g$ is $H$-liftable. In
particular every isomorphism $g$ is liftable.
\end{proof}

\begin{proposition}\label{prop0section-equiv}
For a homomorphism $g\colon H_Y\to H_{Y_1}$ the following are equivalent:
\begin{enumerate}
\item $g$ has a homomorphic section;
\item $g$ is $H$-liftable for every hypergraph $H$;
\item $g$ is $H_{Y_1}$-liftable.
\end{enumerate}
\end{proposition}
\begin{proof}
$(1)\Rightarrow(2)$ is Proposition~\ref{prop:section}. $(2)\Rightarrow(3)$ is
trivial. $(3)\Rightarrow(1)$: apply liftability to
$\psi_1=\Delta_{Y_1}\in\mathrm{Hom}(H_{Y_1},H_{Y_1})$ to get
$s\in\mathrm{Hom}(H_{Y_1},H_Y)$ with $g\circ s=\Delta_{Y_1}$.
\end{proof}

\begin{corollary}\label{cor002}
Let $H_X,H_{X_1}$ be projective planes or affine planes of order $q\ge3$,
viewed as $2$-hypergraphs, and let $H_Y,H_{Y_1}$ be effective with
$2$-definable edges. Let $f\colon X\to X_1$ be a surjection and
$g\colon Y\to Y_1$ a mapping. The following are equivalent:
\begin{enumerate}
\item there is a surjective $\ppi$ with $(f,\ppi,g)$ compatible;
\item $f$ is a collineation and $g$ is an $H_X$-liftable homomorphism;
\item $(f,\ppi,g)$ is an epimorphism of automata for some $\ppi$.
\end{enumerate}
\end{corollary}
\begin{proof}
Let $H$ be a projective plane, or an affine plane of order $q\ge3$. Then
\begin{enumerate}
\item[(A1)] any two points lie on exactly one common line;
\item[(A2)] every line has at least $3$ points: it has $q+1\ge3$ points in the
projective case and $q\ge3$ in the affine case;
\item[(A3)] there are three non-collinear points, i.e.\ a $3$-element set not
contained in any line.
\end{enumerate}
Hence $H$ is a $2$-hypergraph, and it is effective with $2$-definable edges.
Thus $H_X$ and $H_{X_1}$, as well as $H_Y$ and $H_{Y_1}$ by assumption, satisfy
the standing assumptions of Section~\ref{sec:main} with $p=2$.

Condition (P) holds since $p=2\ge2$. Since $H_X$ satisfies (A1) and $H_{X_1}$
satisfies (A3), Lemma~\ref{lem0V1auto} gives a set $V\subseteq X$ with $|V|=3$
such that every $2$-element subset of $V$ is bounded in $H_X$ while $f(V)$ is
unbounded in $H_{X_1}$ (three non-collinear points). So (V) holds for every
surjection $f$, and Theorem~\ref{thm0main} applies. By
Definition~\ref{def00hom}, $f$ is an isomorphism of hypergraphs iff $f$ is a
bijection with $l\in L_X\Leftrightarrow f(l)\in L_{X_1}$, i.e.\ iff $f$ is a
collineation. Replacing ``isomorphism'' by ``collineation'' in
Theorem~\ref{thm0main}(2) gives the claim.

If $H_X$ is projective and $H_{X_1}$ is affine, or vice versa, no collineation
exists: a collineation induces a bijection between the sets of lines and
preserves the property of two lines having a common point, while in a projective
plane any two lines meet and an affine plane has parallel lines. Then (2) fails
for every $f$, and by Theorem~\ref{thm0main} so do (1) and (3).
\end{proof}

\section{Consequences and applications}\label{sec:applications}

Throughout this section all hypergraphs are effective with $p$-definable edges,
and a morphism $\gamma=(f,\ppi,g)$ is compatible in the sense of
Definition~\ref{def0morphism}. Hypotheses (V) and (P) of
Theorem~\ref{thm0main} are in force whenever that theorem is invoked.

Theorem~\ref{thm0main} describes the output component $g$ by liftability. The
next result shows that liftability recovers strong surjectivity on edges under a
mild size condition, and that without liftability the gap between weak and strong
epimorphisms is real.

\begin{proposition}\label{prop0lift1strong}
Let $g\colon H_Y\to H_{Y_1}$ be an $H_X$-liftable homomorphism. Let
$l_1\in L_{Y_1}$ and suppose that
\begin{equation}\label{eq0E}
|e|\ge |l_1|\ \text{ for some } e\in L_X .
\end{equation}
Then $l_1=g(l)$ for some $l\in L_Y$. In particular, if~\eqref{eq0E} holds for
every edge $l_1$ of $H_{Y_1}$, then $g$ is a strong epimorphism.
\end{proposition}
\begin{proof}
Choose $e\in L_X$ with $|e|\ge|l_1|$ and a mapping $\psi_1\colon X\to l_1$ with
$\psi_1(e)=l_1$ (possible since $|e|\ge|l_1|$). As $\psi_1(X)\subseteq l_1$, we
have $\psi_1\in\mathrm{Hom}(H_X,H_{Y_1})$. By liftability there is
$\psi\in\mathrm{Hom}(H_X,H_Y)$ with $g\psi=\psi_1$. Let $l\in L_Y$ with
$\psi(e)\subseteq l$. Then $l_1=\psi_1(e)=g(\psi(e))\subseteq g(l)$.
Since $g$ is a homomorphism, $g(l)\subseteq l_1'$ for some $l_1'\in L_{Y_1}$, so
$l_1\subseteq l_1'$ and $|l_1\cap l_1'|=|l_1|\ge p+1$. By $p$-definability
$l_1=l_1'$. Hence $g(l)=l_1$.

If \eqref{eq0E} holds for every edge of $H_{Y_1}$, then every edge of $H_{Y_1}$
is the image of an edge of $H_Y$; moreover $g$ is a weak epimorphism by
Lemma~\ref{lem0lift1weak}. Thus $g$ is a strong epimorphism.
\end{proof}

For example, let $g=\mathrm{id}\colon H_Y\to H_{Y_1}$ be the identity of
Example~\ref{ex0weaknotstrong}; it is a weak epimorphism that is not strong. It
is also not $H_X$-liftable for any effective $H_X$ with $p=2$. Indeed, take an
edge $e$ of $H_X$ (so $|e|\ge3$) and $\psi_1\colon X\to\{1,4,7\}$ with
$\psi_1(e)=\{1,4,7\}$; then $\psi_1\in\mathrm{Hom}(H_X,H_{Y_1})$, and
$g\psi=\psi_1$ forces $\psi=\psi_1$, but $\psi_1(e)$ is not bounded in $H_Y$.
Hence, by Theorem~\ref{thm0main}, this $g$ is not the output component of an
epimorphism of automata.

The size condition \eqref{eq0E} cannot be dropped altogether.

\begin{example}\label{ex0Eneeded}
Let $p=2$. Let $H_X$ have the vertex set $X=\{x_1,x_2,x_3\}$ and the single edge
$X$. Let $H_{Y_1}$ have the vertex set $Y_1=\{1,2,3,4\}$ and the single edge
$Y_1$. Let $H_Y$ consist of four pairwise disjoint edges
$m_i=\{y_{i,j}: j\in Y_1\setminus\{i\}\}$, $i\in Y_1$, and put $g(y_{i,j})=j$.
All three hypergraphs are effective with $2$-definable edges, and $g$ maps
$m_i$ bijectively onto $Y_1\setminus\{i\}$. Hence $g$ is a homomorphism, it is
surjective on vertices, but $g(m_i)\neq Y_1$ for every $i$, so $g$ is a weak
epimorphism that is not strong. On the other hand $g$ is $H_X$-liftable: every
mapping $\psi_1\colon X\to Y_1$ is a homomorphism (since $Y_1$ is an edge), and
$|\psi_1(X)|\le3$, so $\psi_1(X)\subseteq Y_1\setminus\{i\}$ for some $i$; then
$\psi(x_k)=y_{i,\psi_1(x_k)}$ defines $\psi\colon X\to m_i$, a homomorphism with
$g\psi=\psi_1$. Here $|e|=3<4=|Y_1|$, so \eqref{eq0E} fails.
\end{example}

\begin{proposition}\label{prop0rigidity}
Assume (V) and (P). Let $\gamma=(f,\ppi,g)\colon\mathrm{Atm}(H_X,H_Y)\to
\mathrm{Atm}(H_{X_1},H_{Y_1})$ be an epimorphism. Then:
\begin{enumerate}
\item $f\colon H_X\to H_{X_1}$ is an isomorphism; in particular $|X|=|X_1|$ and
$|L_X|=|L_{X_1}|$;
\item $g$ is a surjective $H_X$-liftable homomorphism, so $|Y_1|\le|Y|$;
\item if $g$ is a strong epimorphism, then $|L_{Y_1}|\le|L_Y|$.
\end{enumerate}
\end{proposition}

\begin{proof}
(1) and (2) follow from Theorem~\ref{thm0main}, $(3)\Rightarrow(2)$, and the
definition of an epimorphism (which includes surjectivity of $g$); an
isomorphism preserves the numbers of vertices and edges. For (3), if every edge
of $H_{Y_1}$ is the image of an edge of $H_Y$, choose for each $l_1\in L_{Y_1}$
an edge $l\in L_Y$ with $g(l)=l_1$; distinct $l_1$ give distinct $l$, so
$|L_{Y_1}|\le|L_Y|$.
\end{proof}

The inequality $|L_{Y_1}|\le|L_Y|$ is not valid for weak epimorphisms: in
Example~\ref{ex0weaknotstrong}, $|L_{Y_1}|=4>3=|L_Y|$. In words: the state
component of an epimorphism of universal hypergraphic automata never identifies
states, and all compression happens on the output side, where $g$ must in
addition be liftable. This contrasts with general automata, where an
epimorphism may identify states. We call the property that $f$ is necessarily an
isomorphism \emph{state rigidity}.

Recall that for a prime power $q=r^{e}$ the projective plane $\mathrm{PG}(2,q)$
has $q^2+q+1$ points and lines of size $q+1$, and its collineation group is
$\mathrm{P\Gamma L}(3,q)$, of order $e\,q^3(q^3-1)(q^2-1)$. The affine plane
$\mathrm{AG}(2,q)$ has $q^2$ points and lines of size $q$. Its collineation group
is $\mathrm{A\Gamma L}(2,q)$, of order $e\,q^2(q^2-1)(q^2-q)$
\cite{hartshorne,karteszi}. Both are $2$-hypergraphs, $\mathrm{PG}(2,q)$ for every
$q$ and $\mathrm{AG}(2,q)$ for $q\ge3$; for
$\mathrm{PG}(2,2)$ axiom (A3) is witnessed by a non-collinear
triple. By Lemma~\ref{lem0V1auto} hypothesis (V) holds automatically and $p=2$
satisfies (P).

\begin{proposition}\label{prop0PG}
Let $H_X=H_{X_1}=\mathrm{PG}(2,q)$ and let $H_Y,H_{Y_1}$ be effective with
$2$-definable edges. Let $f\colon X\to X_1$ be a surjection and $g\colon Y\to Y_1$
a mapping. Then $(f,\ppi,g)$ is an epimorphism of automata for some $\ppi$ if and
only if
\begin{enumerate}
\item $f\in\mathrm{P\Gamma L}(3,q)$ and
\item $g$ is an $H_X$-liftable homomorphism.
\end{enumerate}
In that case, if moreover every edge of $H_{Y_1}$ has at most $q+1$ vertices, then
$g$ is a strong epimorphism. Every $f\in\mathrm{P\Gamma L}(3,q)$ occurs as the
state component of an epimorphism: take $H_{Y_1}=H_Y$ and $g=\Delta_Y$.
\end{proposition}
\begin{proof}
The equivalence is Corollary~\ref{cor002}, since the collineations of
$\mathrm{PG}(2,q)$ are the elements of $\mathrm{P\Gamma L}(3,q)$. The bound on
edge sizes gives \eqref{eq0E} for every edge of $H_{Y_1}$, because the lines of
$H_X$ have $q+1$ points; Proposition~\ref{prop0lift1strong} then applies. For the
last statement, $\Delta_Y$ is an isomorphism, hence liftable by
Proposition~\ref{prop:section}, and Corollary~\ref{cor002} applies.
\end{proof}

\begin{proposition}\label{prop0AG}
Let $H_X=H_{X_1}=\mathrm{AG}(2,q)$ with $q\ge3$. Then the statement of
Proposition~\ref{prop0PG} holds with $\mathrm{P\Gamma L}(3,q)$ replaced by
$\mathrm{A\Gamma L}(2,q)$ and ``$q+1$'' replaced by ``$q$''.
\end{proposition}
For $q=3$, $|\mathrm{P\Gamma L}(3,3)|=3^3\cdot 26\cdot 8=5616$ and
$|\mathrm{A\Gamma L}(2,3)|=3^2\cdot 8\cdot 6=432$. Thus, for every $H_Y$, the
automaton $\mathrm{Atm}(\mathrm{PG}(2,3),H_Y)$ has exactly $5616$ possible state
components of epimorphisms into automata with the same state hypergraph.

Let $L(x_0,F)$ be as in Section~\ref{sec001prelim}, where $x_0\in X$ and
$F\subseteq Y$, and let $L_1(x_1,F_1)$ be defined likewise in
$S(H_{X_1},H_{Y_1})$, for $x_1\in X_1$ and $F_1\subseteq Y_1$.

\begin{proposition}\label{prop0languages}
Let $\gamma=(f,\ppi,g)$ be a compatible triple, $x_0\in X$ and
$F_1\subseteq Y_1$. Then
\begin{equation}\label{eq0preimage}
\ppi^{-1}\bigl(L_1(f(x_0),F_1)\bigr)=L\bigl(x_0,g^{-1}(F_1)\bigr).
\end{equation}
If $\gamma$ is an epimorphism of automata, then
\begin{enumerate}
\item $\ppi\bigl(L(x_0,g^{-1}(F_1))\bigr)=L_1(f(x_0),F_1)$, and every recognisable
language of $\mathrm{Atm}(H_{X_1},H_{Y_1})$ is the image under $\ppi$ of a
recognisable language of $\mathrm{Atm}(H_X,H_Y)$;
\item for every $F\subseteq Y$, $\ppi(L(x_0,F))\subseteq L_1(f(x_0),g(F))$, with
equality when $F=g^{-1}(g(F))$ (i.e.\ $F$ is a union of fibres of $g$).
\end{enumerate}
\end{proposition}
\begin{proof}
For $s=(\varphi,\psi)$ with $\ppi(s)=(\chi,\omega)$, compatibility gives
$\omega(f(x_0))=g(\psi(x_0))$. Hence $\ppi(s)\in L_1(f(x_0),F_1)$ iff
$g(\psi(x_0))\in F_1$ iff $\psi(x_0)\in g^{-1}(F_1)$ iff
$s\in L(x_0,g^{-1}(F_1))$. This is \eqref{eq0preimage}. If $\ppi$ is surjective
then $\ppi(\ppi^{-1}(A))=A$ for every $A$, which gives the first statement of (1);
every initial state of the target has the form $f(x_0)$ since $f$ is surjective,
which gives the second. For (2), $s\in L(x_0,F)$ gives
$\omega(f(x_0))=g(\psi(x_0))\in g(F)$; for $F=g^{-1}g(F)$ apply (1) with
$F_1=g(F)$.
\end{proof}

Thus the transfer of recognisable languages along an epimorphism of automata is
described by \eqref{eq0preimage} and Proposition~\ref{prop0languages}(1).

\section{Conclusion}

We described surjective morphisms between universal hypergraphic automata
$\mathrm{Atm}(H_X,H_Y)$ and between their input symbol semigroups
$S(H_X,H_Y)$ for effective hypergraphs with $p$-definable edges. Under
hypotheses (V) and (P), a triple $(f,P_s,g)$ is an epimorphism of automata if
and only if the state component $f$ is an isomorphism (state rigidity) and the
output component $g$ is a homomorphism which is $H_X$-liftable
(Theorem~\ref{thm0main}). Hypothesis (V) cannot be dropped
(Example~\ref{ex0V1needed}); it holds automatically when $H_X$ satisfies (A1)
and $H_{X_1}$ satisfies (A3) (Lemma~\ref{lem0V1auto}).

Weak and strong epimorphisms of hypergraphs differ
(Example~\ref{ex0weaknotstrong}). Liftability implies weak surjectivity
(Lemma~\ref{lem0lift1weak}) but is not implied even by strong surjectivity
(Example~\ref{ex0notliftable}); it forces strong surjectivity on every edge of
$H_{Y_1}$ not larger than some edge of $H_X$ (Proposition~\ref{prop0lift1strong}),
and this size condition cannot be dropped (Example~\ref{ex0Eneeded}). For
projective and affine planes the state components of epimorphisms are exactly
the collineations (Corollary~\ref{cor002}). Epimorphisms transfer recognisable
languages through preimages (Proposition~\ref{prop0languages}). The results
of~\cite{molchanov2019,khvorostukhina2017,khvorostukhina2025} concern the
determination of $\mathrm{Atm}(H_X,H_Y)$ by $S(H_X,H_Y)$, isomorphisms and
monomorphisms of universal hypergraphic automata; the present paper treats
surjective morphisms under (V) and (P).

\section*{Declarations}
\begin{itemize}
    \item \textbf{Funding:} The author received no external funding for this research.
    \item \textbf{Conflict of interest / Competing interests:} The author declare that they have no conflicts of interest or competing interests related to this study.
    \item \textbf{Ethics approval and consent to participate:} Not applicable.
    \item \textbf{Data availability statement:} All data is included within the manuscript.
\end{itemize}




\begin{thebibliography}{99}

\bibitem{almeida2025}
J.~Almeida, A.~H. Shah,
``Epimorphisms and pseudovarieties of semigroups'',
\emph{Semigroup Forum}, 2025.
DOI: 10.1007/s00233-025-10600-8.

\bibitem{bavel1968}
Z.~Bavel,
``Structure and transition-preserving functions of finite automata'',
\emph{J. ACM}, \textbf{15}:1 (1968), 135--158.
DOI: 10.1145/321439.321448.

\bibitem{bretto}
A.~Bretto,
\emph{Hypergraph Theory: An Introduction},
Springer, Cham, 2013.

\bibitem{cain2009}
A.~J. Cain,
``Automaton semigroups'',
\emph{Theoret. Comput. Sci.}, \textbf{410}:47--49 (2009), 5022--5038.
DOI: 10.1016/j.tcs.2009.07.054.

\bibitem{dangeli2020}
D.~D'Angeli, E.~Rodaro, J.~P. W\"achter,
``On the structure theory of partial automaton semigroups'',
\emph{Semigroup Forum}, \textbf{101}:1 (2020), 51--76.
DOI: 10.1007/s00233-020-10114-5.

\bibitem{eilenberg}
S.~Eilenberg,
\emph{Automata, Languages, and Machines}, vol.~A,
Academic Press, New York, 1974.

\bibitem{hartshorne}
R.~Hartshorne,
\emph{Foundations of Projective Geometry},
Ishi Press, New York, 2009 (reprint of the 1967 edition).

\bibitem{karteszi}
F.~K\'arteszi,
\emph{Introduction to Finite Geometries},
North-Holland, Amsterdam, 1976.

\bibitem{khvorostukhina2025}
E.~V. Khvorostukhina,
``On monomorphisms of universal hypergraphic automata'',
\emph{Lobachevskii J. Math.}, \textbf{46} (2025), 6143--6149.
DOI: 10.1134/S1995080225614109.

\bibitem{khvorostukhina2017}
E.~V. Khvorostukhina, V.~A. Molchanov,
``On problem of concrete characterization of universal automata'',
\emph{Lobachevskii J. Math.}, \textbf{38} (2017), 664--669.

\bibitem{lallement}
G.~Lallement,
\emph{Semigroups and Combinatorial Applications},
Wiley, New York, 1979.

\bibitem{lam1993}
C.~S. Lam,
``Some characterizations of the hereditary pretorsion class of semigroup
automata'',
\emph{Glasgow Math. J.}, \textbf{35}:3 (1993), 327--337.
DOI: 10.1017/s0017089500009903.

\bibitem{molchanov2019}
V.~A. Molchanov, E.~V. Khvorostukhina,
``On problem of abstract definability of universal hypergraphic automata by
input symbol semigroup'',
\emph{Chebyshevskii Sb.}, \textbf{20}:2 (2019), 251--264.

\bibitem{plotkin}
B.~I. Plotkin, L.~Ja. Greenglaz, A.~A. Gvaramija,
\emph{Algebraic Structures in Automata and Databases Theory},
World Scientific, Singapore, 1992.

\bibitem{vagner}
V.~V. Vagner,
\emph{The Theory of Semigroups and Their Applications}, vol.~1,
Saratov University, Saratov, 1965.

\bibitem{yang1975}
C.-C. Yang,
``Structural preserving morphisms of finite automata and an application to
graph isomorphism'',
\emph{IEEE Trans. Comput.}, \textbf{C-24}:11 (1975), 1133--1139.
DOI: 10.1109/t-c.1975.224148.

\bibitem{zhang2023}
C.~Zhang et al.,
``Homotopy and Hom construction in the category of finite hypergraphs'',
\emph{Graphs Combin.}, \textbf{39}:4 (2023).
DOI: 10.1007/s00373-023-02672-6.

\bibitem{zykov}
A.~A. Zykov,
``Hypergraphs'',
\emph{Uspekhi Mat. Nauk}, \textbf{29}:6 (1974), 89--154.

\end{thebibliography}
\end{document}